\documentclass[a4paper,11pt]{amsart}
\usepackage{amsmath,amssymb,amsthm,xcolor,graphicx}
\usepackage[allcolors=purple,citecolor=violet,colorlinks=true]{hyperref}
\usepackage[margin=3cm]{geometry}
\usepackage{marginnote}

\newtheorem{theorem}{{Theorem}}[section]
\newtheorem{proposition}[theorem]{{Proposition}}
\newtheorem{definition}[theorem]{{Definition}}
\newtheorem{lemma}[theorem]{{Lemma}}
\newtheorem{corollary}[theorem]{{Corollary}}
\newtheorem{example}[theorem]{{Example}}

\newcommand{\I}{\mathcal{I}}
\newcommand{\R}{\mathbb{R}}
\newcommand{\dd}{\mathrm{d}}
\newcommand{\vol}{\operatorname{dvol}_g}
\newcommand{\Vol}{\operatorname{vol}_g}
\newcommand{\CS}{\mathcal{S}}
\newcommand{\df}{\mathrm{d}}

\title[Differentiability of the cosmological volume function]{Differentiability and other properties of the cosmological volume function}
\author{Leonardo Garc\'ia-Heveling}
\address{Faculty of Mathematics, University of Vienna \newline \indent Oskar-Morgenstern-Platz 1, 1090 Wien, Austria}
\email{leonardo.garcia.heveling@univie.ac.at}
\urladdr{leogarciaheveling.github.io}
\thanks{I thank Miguel S\'anchez for conversations at early stages of this work, and Greg Galloway and two anonymous referees for detailed feedback on the manuscript.}

\begin{document}

\begin{abstract}
    In a previous work, the regular cosmological volume function $\tau_V$ was introduced as an alternative to the regular cosmological time function of Andersson, Galloway, and Howard. Building on work by Chruściel, Grant and Minguzzi, in this paper we show that in many cases of interest, $\tau_V$ is a continuously differentiable temporal function. This leads to a canonical splitting of the metric tensor, and induces a canonical ``Wick-rotated" Riemannian metric. We also provide some further results and examples related to the cosmological time and volume functions.
\end{abstract}

\maketitle

\section{Introduction}

Two characteristic features of General Relativity are the lack of an observer independent notion of time, and the central role played by singularities. Of these, the second can sometimes relieve the first. For instance, Wald and Yip \cite{WaYi81} proposed that inside of a black hole, one can assign a natural time coordinate to any event via the \emph{maximum lifetime function}, given by the maximum proper time that an observer starting at the event can live, before falling into the singularity. Time-dually, Andersson, Galloway, and Howard \cite{AGH98} proposed to measure time on a cosmological spacetime $(M,g)$ as the maximum proper time that an observer at a given point $p$ can have experienced since the Big Bang. This yields the \emph{cosmological time function}
\begin{equation*}
    \tau(p) := \sup \{ L_g(\gamma) \mid \gamma \text{ a past-directed causal curve starting at } p  \}.
\end{equation*}
Both constructions rely on the relevant spacetime singularity manifesting itself through (future or past) timelike geodesic incompleteness. Other notions of incompleteness exist, such as bounded acceleration incompleteness, and volume incompleteness. The latter was recently introduced by the author in \cite{GH23}, and is defined as the presence of points $p \in M$ with $\Vol(I^\pm(p))$ arbitrarily small (in particular, finite). Here $\Vol$ is the canonical volume measure induced by the metric $g$, and $I^\pm(p)$ the chronological future/past of $p$. On past-volume incomplete spacetimes, it is then natural to define the \emph{cosmological volume function} \cite[Def.~6.1]{GH23}
\begin{equation*}
    \tau_V(p) := \Vol(I^-(p)).
\end{equation*}
The idea of defining time by measuring $I^-(p)$ goes back to Geroch~\cite{Ger70}, and is, by now, standard. In general, however, the canonical volume measure $\Vol$ can be infinite. This is usually circumvented by choosing a different measure (e.g.\ by adding a weight), but the choice of such a finite measure is completely arbitrary. Chru\'sciel, Grant, and Minguzzi~\cite{CGM16} showed that one can always find a (still highly non-unique) weighted measure $\mu$ such that $p \mapsto \mu(I^-(p))$ is a $C^1$ temporal function. Simply choosing the weight to be equal $1$ does not always work, however, even if $\tau_V$ is finite.

In this paper, we identify two settings of interest where $\tau_V$ is a $C^1$ temporal function, and hence a canonical one at that, since its definition only involves notions directly derived from the metric $g$.
\begin{enumerate}
    \item When $(M,g)$ is the future Cauchy development of some initial data: In this case, time is measured from the initial data, rather than from the Big Bang.
    \item A cosmological setting, where $(M,g)$ is past volume incomplete and has no past observer horizons: The second assumption means that $I^+(\gamma) = M$ for every past inextendible causal curve $\gamma$.
\end{enumerate}
See Example~\ref{ex:flrw} for a concrete model where (ii) is satisfied. The precise statements of the results are as follows, and the proofs rely on adapting the arguments in \cite{CGM16}.

\begin{theorem} \label{thm:main}
    Let $(M,g)$ be an orientable spacetime and $\CS$ a future Cauchy surface. Then the cosmological volume function $\tau_V$ of the spacetime $(I^+(\CS),g)$ is regular and is a $C^1$ temporal function.
\end{theorem}

\begin{corollary} \label{cor:main}
 Let $(M,g)$ be an orientable causal spacetime with no past observer horizons, and let its cosmological volume function $\tau_V$ be finite. Then $\tau_V$ is regular and is a $C^1$ temporal function.
\end{corollary}

Note that $\CS$ being future Cauchy implies that it is a codimension $1$ submanifold of class $C^0$, and we are requiring no further differentiability. That $\tau_V$ is $C^1$ implies that $\df \tau_V$ is a continuous one-form field, and we can thus rewrite the metric tensor $g$ as
\begin{equation*}
    g = -\beta (\dd \tau_V)^2 + h,
\end{equation*}
where $\beta$ is a $C^0$ function, and $h$ induces a $C^0$ Riemannian metric on the level sets of $\tau_V$ and vanishes on their $g$-orthogonal complement. We can thus flip the minus sign to obtain the Wick-rotated (or Riemannianized) metric
\begin{equation*}
    g_R := \beta (\dd \tau_V)^2 + h.
\end{equation*}
Continuous metric tensors have been extensively studied. In the Lorentzian case, causal bubbles can appear \cite{ChGr12}, but this is not a concern here, since the causal structure is that of the smooth metric $g$. In the Riemannian case, the usual notions of length and distance are still well-behaved \cite{Bur15}.

One can find many Riemannian metrics on any manifold, and even require properties such as smoothness or completeness, which $g_R$ generally lacks. The important point here is that $g_R$ is constructed only out of the Lorentzian metric $g$, without making any arbitrary choices. This makes it well suited for the study of convergence of sequences $(M_i,g_i)$ of spacetimes via any of the established notions of convergence for metric spaces (e.g.\ Gromov--Hausdorff). A related approach is that of the null distance by Sormani and Vega\footnote{See \cite{SoVe16} and later works. In particular, \cite[Rem.~2.15]{SaSo25} explains that the original idea was to perform Wick rotation with respect to the cosmological time function $\tau$. In general, however, given that $\tau$ is only differentiable almost everywhere, the resulting $g_R$ is also only defined almost everywhere.}. The null distance $\hat{d}_\tau$ also depends on a choice of time function $\tau$, but is able to deal with non-differentiable ones, and can encode the causality of the spacetime in certain cases \cite{SaSo23,BuGH24}. Even if $\tau$ is temporal, $\hat{d}_\tau$ does not coincide with the distance induced by $g_R$, but the two are closely related. Returning to Wick-rotated metric tensors $g_R$, a study of their convergence has recently been performed by Burgos, Flores, and S\'anchez \cite{BFS25}.

\section{Preliminaries}

Throughout the paper, $(M,g)$ will be a $(n+1)$-dimensional oriented and time-oriented Lorentzian manifold. Except where specified otherwise, we assume the metric tensor $g$ to be of regularity $C^{2,1}$, so that all of standard causality theory and the results of \cite{CGM16} are applicable. We denote by $L_g$ and $d_g$ the Lorentzian length and distance, respectively. By $\vol$ or $\vol^{n+1}$ we denote the canonical volume $(n+1)$-form associated to $g$, and by $\Vol(A) := \int_A \vol$ the volume of a set $A \subset M$.

We assume some familiarity with spacetime geometry, but recall the following, perhaps less well-known definition.

\begin{definition} \label{def:fCS}
    A closed, acausal, edgeless subset $\CS \subset M$ is called a \emph{future Cauchy surface} if every past-directed past-inextendible causal curve $\gamma \colon [0,\infty) \to M$ with $p \in I^+(\CS)$ intersects $\CS$ exactly once.
\end{definition}

The defining condition can be equivalently stated as the future Cauchy horizon $H^+(\CS)$ being empty. The usual Cauchy surfaces are a special case (they are exactly those sets which are simultaneously future and past Cauchy). 

Recall that the \emph{cosmological time function} of Andersson, Galloway, and Howard is defined as
\begin{equation*}
 \tau(p) := \sup \{ L_g(\gamma) \mid \gamma \text{ past-directed causal starting at } p\},
\end{equation*}
and called \emph{regular} \cite[Def.~1.1]{AGH98} if it satisfies:
\begin{itemize}
    \item For all points $p \in M$, $\tau(p) < \infty$.
    \item For all past-directed past-inextendible causal curves $\gamma \colon [0,\infty) \to M$, $\tau \circ \gamma(s) \to 0$ as $s \to \infty$.
\end{itemize}
It enjoys the following properties.

\begin{theorem}[{\cite[Thm.~1.2 \& Cor.~2.6]{AGH98}}] \label{thm:timereg}
Suppose $(M,g)$ is a spacetime with regular cosmological time function $\tau$. Then the following properties hold:
\begin{enumerate}
 \item $(M,g)$ is globally hyperbolic.
 \item $\tau$ is continuous and is strictly increasing along future directed causal curves.
 \item The level sets of $\tau$ are future Cauchy.
 \item For each $p\in M$ there is a past-directed unit-speed timelike geodesic $\gamma_p:[0,\tau(p)) \to M$ such that
 \begin{equation*}
  \gamma_p(0)=p,\qquad \tau(\gamma_p(t))=\tau(p) - t,\quad \text{for all } t\in [0,\tau(p)).
 \end{equation*}
 \item The tangent vectors  $\dot\gamma_p(0)$ are locally bounded away from the lightcones. More precisely, if $K\subseteq M$ is compact then $\{\dot\gamma_p(0):p\in K\}$ is compact in the tangent bundle $TM$.
 \item $\tau$ is locally Lipschitz and its first and second derivatives exist almost everywhere.
\end{enumerate}
\end{theorem}

A similar result exists for \emph{regular cosmological volume functions} \cite[Def.~6.1]{GH23},
\begin{equation*}
    \tau_V(p) := \Vol\left(I^-(p)\right),
\end{equation*}
\emph{regularity} again meaning that $\tau_V$ is finite and $\tau_V \to 0$ along all past-inextendible causal curves. On globally hyperbolic spacetimes, the second requirement is redundant.
\begin{proposition}[Prop.~1 in \cite{GaGa25}] \label{prop:1}
    Let $(M,g)$ be a globally hyperbolic spacetime with $\tau_V(p) < \infty$ for every $p \in M$. Then $\tau_V$ is regular.
\end{proposition}

The following parts of Theorem~\ref{thm:timereg} hold also for $\tau_V$.

\begin{theorem}[{\cite[Thms.~6.3 \& 6.4]{GH23}}]
 Let $(M,g)$ be a spacetime with regular cosmological volume function $\tau_V$. Then
 \begin{enumerate}
 \item $(M,g)$ is globally hyperbolic.
 \item $\tau_V$ is continuous and is strictly increasing along future directed causal curves.
 \item The level sets of $\tau_V$ are future Cauchy.
\end{enumerate}
\end{theorem}

Points (i)-(iii) of both theorems are identical. Points (iv) and (v) do not make sense for the cosmological volume function, since $\tau_V(p)$ is, by definition, realized by the set $I^-(p)$. Example \ref{ex:notC1} below has a non-Lipschitz $\tau_V$. Thus (vi) is not true for $\tau_V$ in general, but Theorem~\ref{thm:main} and Corollary~\ref{cor:main} establish that $\tau_V$ is even $C^1$ and temporal in many cases of interest. Recall that any function satisfying (ii) is called a \emph{time function}, while the notion of \emph{temporal function} is stronger, requiring differentiability with past-directed timelike gradient.

\section{Differentiability of the cosmological volume function}

In this section, we prove the main result and some corollaries, and provide related examples.

\subsection{Proof of Theorem~\ref{thm:main}}

Recall that we want to study $\tau_V$ on $I^+(\CS)$ for a Cauchy surface $\CS$, so in other words, the function
\begin{equation*}
    \tau_V(p) = \Vol(I^-(p) \cap I^+(\CS)) = \int_{I^-(p) \cap I^+(\CS)} \vol.
\end{equation*}
We may restrict our spacetime to the interior of the domain of dependence\footnote{See e.g.\ \cite[Chap.~6]{HaEl73} for details on domains of dependence.} of $\CS$, without altering $\tau_V$. This way, without loss of generality, we assume that $(M,g)$ is globally hyperbolic and $\CS$ is a true Cauchy surface.  Regularity of $\tau_V$ is immediate: for every $p \in I^+(\CS)$, the set $I^-(p) \cap I^+(\CS)$ is relatively compact \cite[Prop.~8]{Ger70}, hence of finite volume, and on globally hyperbolic spacetimes, finiteness of $\tau_V$ already implies regularity by Proposition~\ref{prop:1}.

Our proof of differentiability relies heavily on that of the related result for weighted measures by Chruściel, Grant, and Minguzzi \cite{CGM16}. Since the goal is to differentiate an integral, one can guess the expected outcome using the Leibniz integral rule\footnote{This is not exactly how it is presented in \cite{CGM16}. For a reference on the Leibniz integral rule for differential forms, see \cite{Fla73}.}
\begin{equation*}
    \frac{\dd }{\dd s} \biggr\vert_{s=0} \int_{\Omega(s)} \omega(s) = \int_{\Omega(0)} \frac{\dd }{\dd s} \biggr\vert_{s=0} \omega(s) + \int_{\partial \Omega (0)} V \lrcorner \omega(0).
\end{equation*}
Here $V$ denotes the ``velocity" of the boundary. In our case, $\omega = \vol$ is independent of the parameter $s$, so only the boundary term is relevant, and $\Omega(s) = I^-(\gamma(s))$ for some curve $\gamma$. Since $\partial I^-(\gamma(s))$ is ruled by null geodesics emanating from $\gamma(s)$, we may obtain the velocity $V$ as a Jacobi field, given that the solutions of the geodesic equation depend smoothly on initial data. Care has to be taken when dealing with the null cut locus, but since it has measure zero, it ultimately does not contribute\footnote{This is in contrast with the regular cosmological time function $\tau$, where the timelike cut locus leads to lack of differentiability. See also Example~\ref{exam1}.}. In any case, we avoid part of these technical details by referring to \cite{CGM16}, but some work is still required to reduce our problem to their setting.

Let us begin by introducing some notation, following \cite{CGM16}. We fix an auxiliary complete Riemannian metric $h$ on $M$. By $N_h^-I^+(\CS)$, we denote the bundle of past-directed lightlike $h$-unit vectors with basepoint in $I^+(\CS)$, and by $N_h^-(p)$ we denote its fiber over a point $p$ (i.e.\ the $h$-unit past lightcone at $p$). We identify each element $X \in N_h^-I^+(\CS)$ with the unique null geodesic $\gamma$ parametrized by $h$-arclength such that $\dot\gamma(0) = X$. The first step is to ensure that the domain of integration varies continuously with respect to the parameter.

\begin{lemma} \label{lem:tCS}
    The function $t_\CS \colon N_h^- I^+(\CS) \to [0,\infty)$ given by
    \begin{equation}
        t_\CS(\dot\gamma) := \{ t \in [0,\infty) \mid \gamma(t) \in \CS \}
    \end{equation}
    is continuous.
\end{lemma}

\begin{proof}
    Let $\gamma_i \colon [0,\infty) \to M$ be a sequence of past-directed null half-geodesics parametrized by $h$-arclength such that $\gamma_i(0) \to \gamma(0)$ and $\dot\gamma_i(0) \to \dot\gamma(0)$ as $i \to \infty$. First, we show that $t_\CS$ is uniformly bounded above along the sequence $\gamma_i$. To this end, let $p$ be such that $\gamma(0) \in I^-(p)$. Then, since $I^-(p)$ is open, also $\gamma_i(0) \in I^-(p)$ for all $i$ large enough. By global hyperbolicity, the set $K = J^-(p) \cap J^+(\CS)$ is compact, and $(M,g)$ is non-totally imprisoning. Hence, the $h$-length of causal curves in $K$ is uniformly bounded above. This puts an upper bound on $t_\CS$.

    Continuity now follows if we show that for every subsequence $\gamma_{i_j}$, such that $t_j := t_\CS(\gamma_{i_j})$ converges as $j \to \infty$, we have $t_\infty := \lim_{j\to\infty} t_j = t_\CS(\gamma)$. It is easy to see that for such a subsequence $\gamma_{i_j} (t_j) \to \gamma(t_\infty)$ as $j \to \infty$. Since $\CS$ is closed, and by definition of $t_j$, $\gamma_{i_j} (t_j) \in \CS$, also $\gamma(t_\infty) \in \CS$. Therefore $t_\infty = t_\CS(\gamma)$. 
\end{proof}

Next, we define the candidate derivative, and prove that it is continuous. For this, some more notation is needed. Let $E^-(p) := J^-(p) \setminus I^-(p)$. It is well-known that $E^-(p)$ is ruled by past-directed maximizing null geodesics $\gamma$ starting from $p$. For these, the map
\begin{equation*}
    t_-(\gamma) := \sup \{ s \in \R \mid d_g(\gamma(s),\gamma(0)) = 0  \} 
\end{equation*}
indicates until which parameter value $\gamma$ stays length-maximizing. Moreover, we denote by
\begin{equation*}
    \mathring E^-(p) := \{ \gamma(s) \mid \gamma \text{ a past-directed null geodesic, } \gamma(0) = p, s \in (0,t_-(\gamma)) \}
\end{equation*}
the set of all points in $E^-(p)$ that are not a past-endpoint of one of the ruling null geodesics (in other words, we remove the null cut locus). The set $\mathring E^-(p)$ is a $C^{1,1}$ hypersurface, since it is the image under the exponential map $\exp$ of a subset of the past lightcone in $T_pM$, on which $\exp$ is $C^{1,1}$ and injective (see also \cite[p.~2804]{CGM16}).

For a vector $Z \in T_pM$, we define $L(Z)$ to be the vector field on $\mathring E^-(p)$ obtained by solving the Jacobi equation along each null geodesic generator, with initial value $Z$ and vanishing initial derivative. Since the Jacobi equation is linear, $L$ is a linear map. Finally, let $\pi \colon TM \to M$ denote the base-point projection from the tangent bundle to $M$. The map $\hat L$ defined in the next lemma, will turn out to be the differential of $\tau_V$. Here the notation $\lrcorner$ means
\begin{equation*}
    L(Z) \lrcorner \vol = \vol(L(Z), \ \cdot\ ,...,\ \cdot\ )
\end{equation*}
which is to be understood as a $n$-form (hence top-dimensional) on $\mathring E^-(\pi(Z)) \cap I^+(\CS)$.

\begin{lemma} \label{lem:Lcont}
    The map
    \begin{align*}
        \hat L \ \colon \ T I^+(\CS) &\longrightarrow \R \\
        Z &\longmapsto \int_{\mathring E^-(\pi(Z)) \cap I^+(\CS)} L(Z) \lrcorner \vol
    \end{align*}
    is continuous.
\end{lemma}

\begin{proof}
    Note that $\mathring E^-(\pi(Z))$ is ruled by null geodesics, and that the null cut locus has $n$-dimensional Lebesgue measure zero. Hence, up to a set of measure zero, we may parametrize the domain of integration on a subset of $\R \times N^-_h(\pi(Z))$. For each initial velocity in the $h$-unit past lightcone $N^-_h(\pi(Z))$, we need to follow the corresponding null geodesic $\gamma \colon [0,\infty) \to M$ (parametrized by $h$-arclength) until it reaches the null cut locus or the Cauchy surface $\CS$, i.e.\ until parameter value $b(\gamma)$, where
    \begin{equation*}
        b \colon N_h^-I^+(\CS) \to [0,\infty),\quad \gamma \mapsto \min \{ t_-(\gamma), t_\CS (\gamma) \}.
    \end{equation*}
    Since $t_-$ and $t_\CS$ are continuous by \cite[Prop.~2.1]{CGM16} and Lemma~\ref{lem:tCS} respectively, also $b$ is continuous. For some fixed but arbitrary $Z \in TM$, choose a local trivialization of $N^-_h I^+(\CS)$ on a relatively compact neighborhood $U$ of $\pi(Z)$,
    \begin{equation*}
        N^-_h U \longrightarrow \mathbb{S}^{n-1} \times U, \quad \gamma \longmapsto \left(\frac{\dot\gamma(0)}{h_{\gamma(0)} (\dot\gamma(0),\dot\gamma(0))}, \gamma(0)\right).
    \end{equation*}
    We can thus write $b = b(\nu,p)$ as a function of variables $\nu \in \mathbb{S}^{n-1}$ and $p \in U$. From (relative) compactness of $U$ and $\mathbb{S}^{n-1}$, it follows that $b$ is uniformly continuous on $\pi^{-1}(U)$ and bounded above by some constant $C$. Moreover, the solutions $L(Z)$ to the Jacobi equation depend continuously on the initial datum $Z$. We may write the integral above as
    \begin{align*}
        \hat L(Z) &= \int_{\mathbb{S}^{n-1}} \int_0^{b(\nu,\pi(Z))} (L(Z) \lrcorner \vol)_{1...n} \, \df t\, \df \nu \\ &= \int_{\mathbb{S}^{n-1}} \int_0^{C} I_Z(t,\nu) (L(Z) \lrcorner \vol)_{1...n} \, \df t\, \df \nu,
    \end{align*}
    where $I_Z(t,\nu) := 1$ if $t \in (0,b(\nu,\pi(Z))$ and $I_Z(t,\nu) := 0$ otherwise. For any converging sequence $Z_n \to Z$, the integrand on the bottom line converges uniformly, so the integral converges. This proves continuity of $\hat L$.
\end{proof}

Finally, we show that the candidate derivative is, indeed, the derivative.

\begin{lemma}
 The function
 \begin{equation*}
     \tau_V (p) := \int_{I^-(p) \cap I^+(\CS)} \vol
 \end{equation*}
 is a $C^1$ temporal function with differential
 \begin{equation*}
     d \tau_V(Z) = \hat L (Z).
 \end{equation*}
\end{lemma}

\begin{proof}
 Given points $p \in M$, $p' \in I^+(p)$ and a continuous function $f \colon [0,\infty) \to [0,1]$ such that $f(x) = 1$ if $x\leq 1$ and $f(x) = 0$ if $x \geq 2$, we define
 \begin{equation*}
     \phi_\varepsilon(q) := f\left(\frac{d_h\left(q,I^-(p') \cap I^+(\CS)\right)}{\varepsilon}\right).
 \end{equation*}
 Here $d_h$ denotes the distance w.r.t.\ the complete Riemannian metric $h$. The support $\operatorname{supp}(\phi_\varepsilon)$ is compact because
 \begin{equation*}
     \operatorname{supp}(\phi_\varepsilon) \subseteq \left\{ q \in M \mid d_h\left(q,I^-(p') \cap I^+(\CS)\right) \leq 2\varepsilon \right\},
 \end{equation*}
 and $I^-(p') \cap I^+(\CS)$ is compact. Therefore, by \cite[Lem.~3.1]{CGM16}, the functions
    \begin{equation*}
     \tau_{\phi_\varepsilon} (q) := \int_{I^-(q)} {\phi_\varepsilon} \vol
 \end{equation*}
 are $C^1$ with differential
 \begin{equation*}
     \df \tau_{\phi_\varepsilon}(Z) = \int_{\mathring E^-(\pi(Z))} \phi_\varepsilon L(Z) \lrcorner \vol
 \end{equation*}
 (recall also the informal explanation of this fact at the beginning of the section). Hence, for any curve $\gamma$ with $\gamma(0) = p$ and $s_1 < s_2$ small enough so that $\gamma([s_1,s_2]) \subset I^-(p')$, we have
 \begin{equation*}
     \tau_{\phi_\varepsilon} (\gamma(s_2)) - \tau_{\phi_\varepsilon} (\gamma(s_1)) = \int_{s_1}^{s_2} \left( \int_{\mathring E^-(\gamma(s))} \phi_\varepsilon L(\dot \gamma(s)) \lrcorner \vol \right)\df s.
 \end{equation*}
 Take the limit $\varepsilon \to 0$ on both sides, noticing that the $\phi_\varepsilon$ converge to the indicator function of $I^-(p') \cap I^+(\CS)$ and that $\mathring E^-(\gamma(s)) \subset I^-(p')$, to obtain
 \begin{equation} \label{eq:difference}
     \tau_V (\gamma(s_2)) - \tau_V (\gamma(s_1)) = \int_{s_1}^{s_2} \left( \int_{\mathring E^-(\gamma(s)) \cap I^+(\CS)} L(\dot \gamma(s)) \lrcorner \vol \right) \df s.
 \end{equation}
 By Lemma~\ref{lem:Lcont}, the integrand is continuous in $s$, and the result now follows from the fundamental theorem of Calculus.
 
 Finally, we prove that $\df \tau_V$ is timelike (in particular, non-vanishing). This can be done along the same lines as in \cite[p.~2810]{CGM16}: From \eqref{eq:difference} it follows that for $X \in T_pM$,
 \begin{equation} \label{eq:dtauV}
     \df \tau_V (X) = \int_{\mathring E^-(p) \cap I^+(\CS)} L(X) \lrcorner \vol.
 \end{equation}
 It suffices to show that this expression is positive for all $X$ future-directed causal. Locally on $\mathring E^-(X) \cap I^+(\CS)$, we may choose double null coordinates $u,v$ such that $\partial_u$ is tangent to the null generators of $\mathring E^-(\gamma(s))$ and
 \begin{equation*}
     g = -du dv + \bar g
 \end{equation*}
 where $\bar{g}$ is positive-definite on the subbundle of $T \mathring E^-(p)$ orthogonal to $\partial_u$. Then
 \begin{equation*}
     \vol = 2 \df u \df v \df A
 \end{equation*}
 for $\df A$ the area element associated to $\bar{g}$, and hence
 \begin{equation*}
     L(X) \lrcorner \vol = 2 \df v(L(X)) \df u \df A = -g(L(X),\partial_u) \df u \df A.
 \end{equation*}
 Along each null generator of $\mathring E^-(p)$, $L(X)$ is a Jacobi field and hence $g(L(X),\partial_u)$ is constant. Moreover, (for each generator separately) the tangent vector $\partial_u$ can be extended to $p$, where $L(X) = X$, and since both $X$ and $\partial_u$ are future-directed causal,
 \begin{equation*}
     g(L(X),\partial_u) = g_p(X,\partial_u \vert_p) \leq 0,
 \end{equation*}
 where equality can only happen on at most one of the generators (the one tangent to $X$ if $X$ is null). This proves that the integral in \eqref{eq:dtauV} is positive and hence $\df \tau_V(X) > 0$.
\end{proof}

\subsection{Corollaries}

The first corollary implies the ones further below, and is also interesting on its own.

\begin{corollary} \label{cor:2}
    Let $(M,g)$ be a spacetime with regular cosmological volume function $\tau_V$. If $p \in M$ is a point such that $I^-(p)$ contains a future Cauchy surface $\CS$, then $\tau_V$ is $C^1$ with timelike gradient in a neighborhood of $p$.
\end{corollary}

\begin{proof}
    Since $I^-(p) \cap \CS = \CS$ is compact, and chronological pasts are open, it is easy to infer the existence of a neighborhood $U$ of $p$ such that $\CS \subset I^-(q)$ for every $q \in U$. Then
    \begin{equation*}
        I^-(q) = \left(I^-(q) \cap I^+(\CS)\right) \cup \CS \cup I^-(\CS),
    \end{equation*}
    which can be seen as follows: For each $x \in I^-(q)$, let $\gamma$ be the timelike curve $\gamma$ joining $x$ and $q$. Either $\gamma$ intersects $\CS$, in which case $x \in \CS \cup I^-(\CS)$, or it doesn't, in which case $x \in I^+(\CS)$, because $\gamma$ can be extended to the past until it intersects $\CS$. Then
    \begin{equation*}
        \tau_V(q) = \Vol\left(I^-(q) \cap I^+(\CS)\right) + \Vol\left(I^-(\CS)\right),
    \end{equation*}
    where the first summand is $C^1$ by Theorem~\ref{thm:main}, and the second summand is constant.
\end{proof}

In order to prove Corollary \ref{cor:main}, we show that every point has an entire Cauchy surface in its past. This is a consequence of the characterization of the No Past Observer Horizons (NPOH) condition
\begin{equation*}
     I^+(\gamma) = M \ \text{ for every past-inextendible causal curve } \gamma.
\end{equation*}
given in \cite{GaZe}, which we rephrase here for the current purpose.

\begin{lemma}[Prop.~2.7 and Lem. 2.9 in \cite{GaZe}] \label{lem:NPOH}
    Let $(M,g)$ be a causal spacetime satisfying the  NPOH. Then, $(M,g)$ is globally hyperbolic with compact Cauchy surfaces. Moreover, for any Cauchy temporal function $\tau_C \colon M \to \R$, and any $t \in \R$, there is $t' < t$ such that:
    \begin{itemize}
        \item For any $p$ with $\tau_C(p) = t$, $\tau_C^{-1}(t') \subset I^-(p)$.
        \item For any $p'$ with $\tau_C(p') = t'$, $\tau_C^{-1}(t) \subset I^+(p')$.
    \end{itemize}
\end{lemma}

Now let $(M,g)$ be a causal spacetime with finite $\tau_V$ and satisfying the NPOH (as in Corollary~\ref{cor:main}), and let $p \in M$ be any point. Then $(M,g)$ is globally hyperbolic, and we may choose a Cauchy temporal function $\tau_C$ by \cite{BeSa05}. Applying Lemma~\ref{lem:NPOH} with $t = \tau_C(p)$, we conclude that the Cauchy surface $\tau^{-1}(p')$ is contained in $I^-(p)$. Moreover, by Proposition~\ref{prop:1}, global hyperbolicity and finiteness of $\tau_V$ together imply regularity of $\tau_V$. Thus Corollary~\ref{cor:2} is applicable, and since $p$ was arbitrary, this proves Corollary~\ref{cor:main}.

Time-reversing our assumption to the No Future Observer Horizons (NFOH) condition,
\begin{equation*}
     I^-(\gamma) = M \ \text{ for every future-inextendible causal curve } \gamma,
\end{equation*}
but keeping the definition of $\tau_V(\cdot) = \vol(I^-(\cdot))$ intact, we obtain a result similar to Corollary~\ref{cor:main}, but only from a certain time on.

\begin{corollary}
    Let $(M,g)$ be a causal spacetime satisfying the NFOH, and suppose that $\tau_V$ is finite. Then, there is a compact Cauchy surface $\CS \subset M$ such that $\tau_V$ is $C^1$ with timelike gradient at every point in $I^+(\CS)$. In particular, this includes every point $p$ with
    \begin{equation*}
        \tau_V(p) > \max\{\tau_V(q) \mid q \in \CS \}.
    \end{equation*}
\end{corollary}

Note that here we mean the cosmological volume function $\tau_V$ of the spacetime $(M,g)$, unlike in Theorem~\ref{thm:main}, where we meant the $\tau_V$ of the spacetime $(I^+(\CS),g)$ .

\begin{proof}
    By the time-reverse of Lemma \ref{lem:NPOH}, $(M,g)$ is globally hyperbolic with compact Cauchy surfaces, hence has regular $\tau_V$ by Proposition~\ref{prop:1} and admits a Cauchy temporal function $\tau_C$ by \cite{BeSa05}. Let $t \in \R$ be any value, and $\CS := \tau_C^{-1}(t)$. Again by the time-reverse of Lemma~\ref{lem:NPOH}, there is a $t' > t$ such that $\CS \subset I^-(p)$ for all $p$ with $\tau_C(p) = t'$, and hence also for all $p$ with $\tau_C(p) \geq t'$. Corollary~\ref{cor:2} then implies that $\tau_V$ is $C^1$ at $p$, with past-directed timelike gradient.
\end{proof}

\subsection{Examples regarding differentiability of $\tau_V$}

We start with an example to motivate the relevance of Corollary~\ref{cor:main} in the cosmological context.

\begin{example}[Corollary \ref{cor:main} applicable] \label{ex:flrw}
 Let $(N,h)$ be a compact Riemannian $n$-manifold, and let $a \colon (0,\infty) \to (0,\infty)$ be a smooth function. Consider the Lorentzian warped product
 \begin{equation*}
     M := (0,\infty) \times N, \qquad g := -dt^2 + a(t)^2 h,
 \end{equation*}
 which is globally hyperbolic. In particular, this is an FLRW spacetime if $(N,h)$ is the round sphere $\mathbb{S}^n$ (or a quotient of an FLRW if $(N,h)$ is a compact quotient of $\R^n$ or $\mathbb{H}^n$). Suppose that $(M,g)$ has no particle horizons, which by definition means that
 \begin{equation*}
     \int_0^1 \frac{1}{a(t)} \dd t = \infty.
 \end{equation*}
 Then $(M,g)$ satisfies the NPOH condition. A proof of this implication can be found in \cite[Prop.~3.6]{Sbi25}, stated there for $(N,h) = \mathbb{S}^n$, but valid for any compact manifold. Moreover, since $I^-(p) \subsetneq (0,t(p)) \times N$,
 \begin{equation*}
     \tau_V(p) <  \operatorname{vol}_h(N) \int_0^{t(p)} a^n(t) \dd t,
 \end{equation*}
 so if the right hand side is finite (for example, if $a(t) \to 0$ as $t \to 0$), then $\tau_V$ is finite and, by Proposition~\ref{prop:1}, regular.
\end{example}

The next example is very similar to that of \cite[Fig.~1]{CGM16}, and shows that if we allow $\CS$ to have lightlike portions, then Theorem~\ref{thm:main} no longer holds.

\begin{example}[$\tau_V$ not $C^1$] \label{ex:vol}
 Let $f$ be the function
 \begin{equation*}
  f(x) := \begin{cases}
        0 &\text{if } x < 0,\\
        x &\text{if } 0 \leq x \leq 1,\\
        1 &\text{if } 1 \leq x,
       \end{cases}
 \end{equation*}
 and let $M := \{ (t,x) \in \R^{1,1} \mid t > f(x) \}$ be equipped with the Minkowski metric $-\df t^2 + \df x^2$. Then $\tau_V$ is regular, but it is not differentiable along $\{t=x\}$. To illustrate this, we consider $p := (2,2)$ and compute the directional derivatives from the left and from the right along the direction $v := (-1,1)$ (see also Figure~\ref{fig:cosmo4}). For $\varepsilon > 0$, we have
 \begin{align*}
  \tau_V (p - \varepsilon v) &= (1-\varepsilon)^2 \\
  \tau_V (p + \varepsilon v) &= (1+\varepsilon)^2 + 2\varepsilon
 \end{align*}
 and hence
 \begin{equation*}
  D^-_{v} \tau_V = 2 \neq 4 = D^+_v \tau_V.
 \end{equation*}
\end{example}

\begin{figure}
 \centering
 \includegraphics{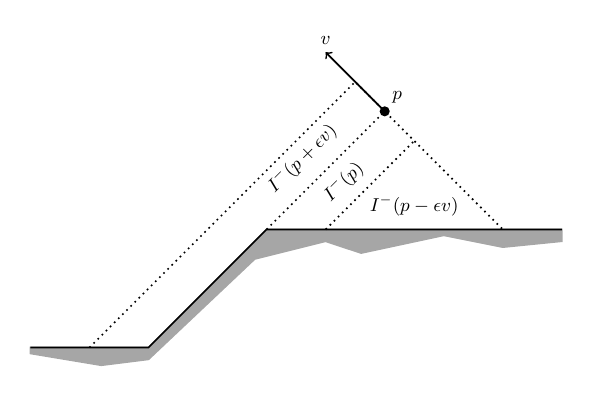}
 \caption{A sketch of the spacetime in Example~\ref{ex:vol}, with the past sets of $p \pm \varepsilon v$ indicated (whose volume equals the cosmological volume function at $p \pm \varepsilon v$).}
 \label{fig:cosmo4}
\end{figure}

Next we present a spacetime where $\tau_V$ fails to be locally Lipschitz (cf.\ Theorem~\ref{thm:timereg}(vi)). Unlike the previous example, this one has a spacelike causal boundary (consisting of more than one point, so Corollary~\ref{cor:main} is not applicable).

\begin{example}[$\tau_V$ not $C^{0,1}$] \label{ex:notC1}
    Let $M := (0,\infty) \times \R$ and $g := \Omega (-\df t^2 + \df x^2)$, where
    \begin{equation*}
        \Omega(t,x) = \begin{cases}
            (t^2+x^2)^{-3/2} & \text{if } x < -t, \\
            1 & \text{if } x>0, \\
            \text{smooth} & \text{else.}
        \end{cases}
    \end{equation*}
    We show that $\tau_V$ is not differentiable at $(1,1)$, in fact not even locally Lipschitz. For any $\varepsilon > 0$, $I^-((1+\varepsilon,1))$ intersects the wedge $A := \{ x < -t \}$, which gives a finite but non-Lipschitz contribution to the volume due to the choice of conformal factor $\Omega$. In particular, for $r = \varepsilon / \sqrt{2}$, we have
    \begin{equation*}
        \left\{ \sqrt{t^2 + x^2} > r \right\} \cap A \subset I^-\left((1+\varepsilon,1)\right) \setminus I^-((1,1)),
    \end{equation*}
    as depicted in Figure~\ref{fig:proof}, so
    \begin{align*}
         \tau_V((1+\varepsilon,1)) - \tau_V((1,1)) &> \Vol\left(\left\{ \sqrt{t^2 + x^2}  r \right\} \cap A\right) \\ &= \frac{\pi}{4} \int_0^{r} \frac{1}{\sqrt{\rho}} \,\df \rho = \frac{\pi}{4} \sqrt{\frac{\varepsilon}{\sqrt{2}}}.
    \end{align*}
    But then
    \begin{equation*}
        \frac{\tau_V((1+\varepsilon,1)) - \tau_V((1,1))}{\varepsilon} \to \infty \quad \text{as} \quad \varepsilon \to 0,
    \end{equation*}
    so $\tau_V$ fails the Lipschitz property in any neighborhood of $(1,1)$. Note that the metric $g$ cannot be extended beyond $t=0$ (at least not in the given coordinates), given that $\Omega \to \infty$ when approaching $(0,0)$ from within $A$. Hence $(M,g)$ cannot be cast in the form $I^+(\CS)$, for $\CS$ a Cauchy surface in a larger spacetime, and Theorem~\ref{thm:main} is not applicable.
\end{example}

\begin{figure}
 \centering
 \includegraphics[scale=1.2]{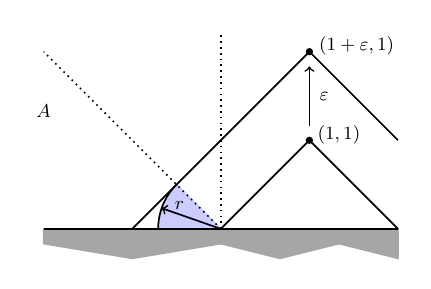}
 \caption{A sketch of the argument that shows lack of differentiability in Example~\ref{ex:notC1}.}
 \label{fig:proof}
\end{figure}

Finally, we give an example of a spacetime where Theorem~\ref{thm:main} is applicable, and $\tau_V$ is $C^1$ but not $C^2$.

\begin{example}[$\tau_V$ not $C^2$] \label{ex:notC2}
    Let $M := (0,\infty) \times S^1$ and $g := -\df t^2 + \df \theta^2$ (we can view $M$ as $I^+(\{t=0\})$ in the larger spacetime $(\R \times S^1,g)$). See Figure~\ref{fig:notC2} for a sketch of the following argument: The past of a point $(t,\theta)$ with $0 < t \leq \pi$ is an isosceles triangle of height $t$ and width $2t$, so
    \begin{equation*}
        \tau_V(t,\theta) = t^2, \quad \df \tau_V = 2 t \df t, \quad \frac{\partial^2 \tau_V}{\partial t^2} = 2.
    \end{equation*}
    If, on the other hand, $t \geq  \pi$, then $I^-((t,\theta))$ is the union of an isosceles triangle of height $\pi$ and width $2\pi$ and a rectangle of height $t-\pi$ and width $2\pi$, so
    \begin{equation*}
        \tau_V(t,x) = 2\pi (t-\pi) + \pi^2, \quad \df \tau_V = 2 \pi \df t, \quad \frac{\partial^2 \tau_V}{\partial t^2} = 0.
    \end{equation*}
    We observe that the two expressions for $\df \tau_V$ match continuously at $t=\pi$, but the ones for the second derivative do not.
\end{example}

\begin{figure}
 \centering
 \includegraphics[scale=1.25]{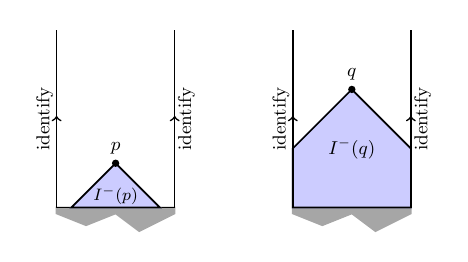} 
 \caption{A sketch of Example~\ref{ex:notC2}. On the left, a point with $t < \pi$. On the right, a point with $t > \pi$.} 
 \label{fig:notC2}
\end{figure}

\section{Comparison of cosmological time and volume functions}

In this section, we provide further results and examples to compare the properties of $\tau$ and $\tau_V$.

\subsection{Restriction to subsets}

The cosmological time function $\tau$ behaves well with respect to restriction to certain subsets of $M$. In particular, if $\CS = \{\tau= c\}$ for some $c > 0$, and $\tilde\tau$ is the cosmological time function of $I^+(\CS)$ (viewed as a spacetime of its own), then it is easy to prove that $\tilde\tau$ is also regular, and
\begin{equation*}
    \tilde\tau(p) = \tau(p) - c \quad \text{for all } p \in I^+(\CS).
\end{equation*}
An analogous result does not exist for the cosmological volume function, since the portion of the volume of $I^-(p) \cap J^-(\CS)$ varies with $p$ even if $p \in I^+(\CS)$ (except in the special case when $\CS \subset I^-(p)$, which is treated in Corollary~\ref{cor:2}).

\subsection{Regularity}

Recall that regularity of the cosmological time or volume function means that it is finite and tends to zero along all past-inextendible causal curves. It is a crucial property needed to ensure that these are even time functions in the usual sense (i.e.\ continuous and strictly increasing along all future-directed causal curves). Using volume comparison techniques recently developed by Cavalletti and Mondino \cite{CaMo24}, we identify conditions such that regularity of $\tau$ implies regularity of $\tau_V$.

\begin{proposition}
 Let $(M,g)$ have a regular cosmological time function $\tau$, compact Cauchy surfaces, and $\operatorname{Ric} \geq -\kappa g$ in timelike directions, for some constant $\kappa \in \R$. Then $(M,g)$ has a regular cosmological volume function $\tau_V$.
\end{proposition}

\begin{proof}
 By Proposition~\ref{prop:Cauchy} below, there is some $b > 0$ such that $\CS := \{ p \mid \tau(p) = b\}$ and $\CS' := \{ p  \mid \tau(p) = b/2\}$ are Cauchy surfaces. Let $\tilde{\CS}_{t} := \{p  \mid d(p,\CS) = t \}$. The function $d(\cdot,\CS)$ is continuous (see \cite[Lem.~3.4]{GaVe14}, time-reversed), and hence reaches a minimum value $a > 0$ on the compact set $\CS'$. Thus $D := \{ p  \mid 0 \leq d(p,\CS) \leq a \}$ is contained between $\CS$ and $\CS'$, and is compact. Denote by $\Vol^n$ the $n$-dimensional measure of a hypersurface. We apply the coarea formula to obtain
 \begin{equation*} 
  \Vol(I^-(\CS)) = \int_{I^-(\CS)} \vol^{n+1} =  \int_0^a \operatorname{vol^n_g}(\tilde\CS_t) \, \df t + \int_a^b \operatorname{vol^n_g}(\tilde\CS_t) \, \df t.
 \end{equation*}
 The first term on the right hand side is finite, because it is simply the $n+1$-dimensional volume of the compact region $D$. To show that also the second term is finite, we use \cite[Thm.~5.1]{CaMo24}, which states that
 \begin{equation*}
     t \longmapsto \frac{\Vol^n(\tilde\CS_t)}{\mathfrak{s}_\kappa(t)}
 \end{equation*}
 is non-increasing, for
 \begin{equation*}
    \mathfrak{s}_{\kappa}(t):=
    \begin{cases}
    \frac{1}{\sqrt{\kappa}} \sin(\sqrt{\kappa} t)   &\text{if } \kappa>0,\\
    t &\text{if } \kappa=0,\\
    \frac{1}{\sqrt{-\kappa}} \sinh(\sqrt{-\kappa} t)   &\text{if } \kappa<0.\\
    \end{cases}
 \end{equation*}
 It follows that
 \begin{equation*}
     \int_a^b \operatorname{vol^n_g}(\tilde\CS_t) \, \df t \leq \frac{\Vol^n(\tilde \CS_a)}{\mathfrak{s}_{\kappa}(a)} \int_a^b \mathfrak{s}_\kappa(t) \, \df t
 \end{equation*}
 is finite. Hence $\Vol(I^-(\CS)) < \infty$. Now since $\CS$ is Cauchy, $M = I^-(\CS) \cup \CS \cup I^+(\CS)$. For $p \in I^-(\CS) \cup \CS$, we have $I^-(p) \subset I^-(\CS)$ and hence $\tau_V(p) < \infty$. For $p \in I^+(\CS)$, we have
 \begin{equation*}
     I^-(p) \subseteq I^-(\CS) \cup \CS \cup \left( I^-(p) \cap I^+(\CS) \right).
 \end{equation*}
 Here the set in brackets is compact (because $\CS$ is Cauchy), and hence has finite volume. It follows that $\tau_V(p) < \infty$. We conclude that $\tau_V$ is finite, and since $(M,g)$ is globally hyperbolic, this implies regularity by Proposition~\ref{prop:1}.
\end{proof}

\subsection{Cauchy level sets}

Theorem~\ref{thm:timereg}(iii) establishes that the level sets $\CS_t := \{\tau = t\}$ of the cosmological time function are future Cauchy (Definition~\ref{def:fCS}). The question of when they are Cauchy was first investigated in \cite{GaGa25}.

\begin{theorem}[{\cite[Thm.~5]{GaGa25}}] \label{thm:GaGa25}
 Let $(M,g)$ be a future timelike geodesically complete spacetime satisfying at least one of the following:
 \begin{enumerate}
  \item $(M,g)$ contains a compact Cauchy surface.
  \item The future causal boundary of $(M,g)$ is spacelike.
 \end{enumerate}
 If the cosmological time function $\tau \colon M \to (0,\infty)$ is regular, then its level sets are Cauchy surfaces.
\end{theorem}

It remained open if assumptions (i)-(ii) can be removed. Example~\ref{counterexCauchy} below gives a negative answer. Example 1 in \cite{GaGa25} illustrates that future timelike completeness can also not be removed. The next proposition, however, shows that if we are only interested in small times, then spatial compactness alone suffices. This is also true for the cosmological volume function. The proof is adapted from that of \cite[Thm.~4.2]{GaZe}.

\begin{proposition} \label{prop:Cauchy}
 If $(M,g)$ is globally hyperbolic with compact Cauchy surfaces and admits a regular cosmological time or volume function $\tau$. Then there exists $a > 0$ such that $\{\tau = t\}$ is a Cauchy surface for every $0 < t \leq a$.
\end{proposition}

\begin{proof}
 Let $S$ be a Cauchy surface. By compactness, $\tau$ achieves a minimum value $a$ on $S$. Since every inextendible causal curve, say $\gamma \colon I \to M$, intersects $S$, there is $s \in I$ such that $\tau \circ \gamma(s) \geq a$. By regularity, $\tau \circ \gamma \to 0$, and it follows that all values in $(0,a]$ are achieved. Hence the corresponding $\tau$-level sets are Cauchy surfaces.
\end{proof}

The following example is adapted from a very similar one in \cite[Sec.~4.3]{AGH98}, where it is used to illustrate that $\tau < \infty$ and $\tau \to 0$ along timelike\footnote{Some sentences in \cite{AGH98} read ``causal", but ``timelike" is really meant (confirmed in private communication with G.~Galloway).} geodesics is not a sufficient condition for regularity of $\tau$. Our modified example has the properties of being future-timelike geodesically complete and admitting a regular cosmological time function whose level sets, however, are not Cauchy.

\begin{example} \label{counterexCauchy}
 Let $\R^3$ be equipped with the $C^1$ Lorentzian metric
 \begin{equation*}
  g := \df y^2 + e^{2y} \left( \df x \df t + (\vert t \vert^{2\alpha} + f(y)) \df x^2 \right),
 \end{equation*}
 where $\frac{1}{2} < \alpha < 1$ and $f(y) := e^{y^2} - 1$. We choose the time orientation so that $\dot t \geq 0$ along all future-directed causal curves. The curve $\eta$ with image $\{t=0,y=0\}$ is a null geodesic (see Figure~\ref{fig:example}).
 
 Suppose that $\gamma \colon [0,\infty) \to \R^3,\ s \mapsto (t(s),x(s),y(s))$ is a future-inextendible future-directed causal curve. We prove that either $\gamma$ is asymptotic to $\eta$, or else $t \to \infty$. The form of the metric $g$ implies that
 \begin{equation*}
     \dot x \leq 0 \text{ and } \dot t \geq -(\vert t \vert^{2 \alpha} + f(y))\dot x.
 \end{equation*}
 Hence, if $y \to \pm \infty$ as $s \to \infty$, then $\dot t \to \infty$, and since the range of $s$ is infinite, $t \to \infty$. Moreover, if $x \to -\infty$ ($+\infty$ is not possible by the choice of time orientation), then either $t\to 0$ and $y\to 0$, or else $t \to \infty$. This proves the claim, which by symmetry is also true for past-inextendible past-directed causal curves.

 It follows from the asymptotic properties of causal curves that $\CS = \{ t = \max\{-1,x\}\}$ is a future Cauchy surface. Hence on $M := I^+(\CS)$ the cosmological time function takes the form $\tau(p) = d(\CS,p)$ and is regular. Moreover, by the time-reverse of \cite[Lem.~4.2]{AGH98}, there is a uniform upper bound on the length of any future-directed causal curve in $M$ with endpoint on $\eta$, so, in particular, $\sup (\tau \circ \eta) < \infty$. Hence the $\tau$-level sets for values larger than $\sup (\tau \circ \eta)$ are not Cauchy surfaces, given that $\eta$ is an inextendible causal curve that does not intersect them.

 It remains to prove that $(M,g)$ can be modified to be future timelike geodesically complete ($(M,g)$ in the current form might already have this property, but that would be more difficult to prove). For this purpose, we show that future-inextendible timelike geodesics always have $t \to +\infty$ (i.e.~they cannot asymptote to $\eta$). Since $\partial_x$ is a Killing vector field, we obtain a constant of motion
 \begin{equation*}
     C = g(\dot\gamma,\partial_x) = e^{2y} \left(\dot t + \left( \vert t \vert^{2 \alpha} + f(y) \right) \dot x \right).
 \end{equation*}
 For a future-directed timelike geodesic $\gamma$, necessarily $C > 0$ and $\dot x < 0$. Hence
 \begin{equation*}
     \dot t = C e^{-2y} - \left( \vert t \vert^{2 \alpha} + f(y) \right) \dot x, 
 \end{equation*}
 Parametrizing $\gamma$ to unit-speed, we furthermore obtain
 \begin{equation*}
     C \dot x + \dot y^2 = -1.
 \end{equation*}
 It follows that $\vert \dot x \vert \geq C^{-1}$, and combining this with the previous equation, we conclude that $\dot t \geq C'$ for some constant $C'>0$. If the parameter range $I$ of $\gamma$ is unbounded, it readily follows that $t \to \infty$. In the case that $I$ is bounded, one can show by integrating the above equations that $t \to \infty$ if $x$ or $y$ are unbounded. But if $x$ and $y$ are bounded, then we also must have $t \to \infty$, given that $\gamma$ was assumed inextendible. This proves that $t \to \infty$ along all future-inextendible future-directed timelike geodesics. Now we can modify $(M,g)$ on the set $\{t > 1\}$ to ensure that these geodesics are complete (e.g.~with a conformal factor, or by gluing in a portion of a spacetime known to be complete). Modifying only on $\{t>1\}$ ensures that the behavior of $\tau$ along $\eta$ is not altered.
\end{example}

\begin{figure}
 \centering
 \includegraphics[scale=1.25]{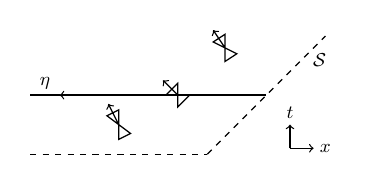}
 \caption{A sketch of the subset $\{y=0\}$ in Example~\ref{counterexCauchy}. Note that the lightcones are narrower away from $\eta$.}
 \label{fig:example}
\end{figure}

One can construct similar examples with infinitely many curves $\eta_i$, such that for any $\varepsilon > 0$, there exists $k$ with $\sup (\tau \circ \eta_k) < \varepsilon$. Then none of the $\tau$-level sets is Cauchy.

\subsection{Cosmological cut locus}

Given a smooth, spacelike Cauchy surface $\CS$, the Lorentzian distance to $\CS$ coincides with the cosmological time function of the spacetime $I^+(\CS)$, and is regular. Hence $\tau$ is smooth in a vicinity of $\CS$, to be precise, away from the cut locus. See \cite{GrSo22} for further discussion of cosmological time functions that arise in this way. Note also that Theorem~\ref{thm:main} deals with the cosmological volume function in the same setting.

Now, let $(M,g)$ be an arbitrary spacetime with regular cosmological time $\tau$. One can try to think of $\tau$ as the distance to some ``initial singularity". In particular, one may define a set analogous to the cut locus to a hypersurface, as follows.

\begin{definition} \label{def:I}
 We define the subset $\I \subset M$ by declaring that $p \in \I$ if there is $\varepsilon > 0$ and a past-directed unit-speed timelike geodesic $\gamma_p \colon (-\varepsilon, \tau(p)) \to M$ with $\gamma_p(0) = p$ such that $\tau(\gamma(s)) = \tau(p) - s$ for all $s \in (-\varepsilon, \tau(p))$. We define the \emph{cosmological cut locus} as $\mathcal{C} := M \setminus \I$.
\end{definition}

Theorem~\ref{thm:timereg}(iv) tells us that a geodesic $\gamma_p$ as in the definition can always be defined on $[0,\tau(p))$, so the non-trivial requirement for $p \in \I$ is that $\gamma_p$ can be prolonged for some short time $\varepsilon$ while retaining the maximization property. The following facts are easy to establish:
\begin{itemize}
 \item $\I$ is dense in $M$.
 \item If $p \in \I$, then the geodesic $\gamma_p \colon [0,\tau(p)) \to M$ provided by Theorem~\ref{thm:timereg}(iv) is unique.
\end{itemize}
To prove the first statement, take any point $p \in M$, and the curve $\gamma_p$ given by Theorem~\ref{thm:timereg}(iv). Then $\gamma_p((0,\tau(p)) \subset \I$, so $p \in \overline{\I}$. The second statement is a direct consequence of the fact that geodesics in smooth spacetimes cannot branch. A related result proven by Sormani and Vega \cite[Prop.~5.3]{SoVe16} is the following:
\begin{itemize}
    \item If $\tau$ is differentiable at $p$, then its gradient vector equals $\dot \gamma_p(0)$ for $\gamma_p$ as in Theorem~\ref{thm:timereg}(iv), and moreover, $\gamma_p$ is unique.
\end{itemize}
In particular, it follows that the gradient of $\tau$ is past-directed unit timelike wherever it exists. Note that the gradient of $\tau_V$ is also past-directed timelike where defined (see proof of Theorem~\ref{thm:main} and \cite[p.~2810]{CGM16}), but it will not be unit in most cases.

The above discussion raises two natural questions:
\begin{itemize}
    \item Is there an implication between $p \in \I$ and differentiability of $\tau$ at $p$?
    \item Is $\mathcal{C}$ closed?
\end{itemize}
Recall that the usual cut locus is closed. We leave the first question open, and give a negative answer to the second one, via Example~\ref{exam2}. We first give a preliminary example where $\I \neq M$.

\begin{example} \label{exam1}
 Consider a function of one variable of the form
 \begin{equation*}
  f(x) := \begin{cases}
             2-\sqrt{1+x^2} &\text{if } 0 \leq \vert x \vert \leq 1, \\
             0 &\text{if } 2 \leq \vert x \vert, \\
             \text{smoothly interpolated} &\text{else.}
            \end{cases}
 \end{equation*}
 Let $M := \{ (t,x) \in \R^{1,1} \mid t > f(x) \}$ be equipped with the Minkowski metric $-\df t^2 + \df x^2$. The graph of $f$ for $0 \leq \vert x \vert \leq 1$ coincides with the hyperboloid 
 \begin{equation*}
     H := \{q \in \R^{1,1} \mid d(q,p) = 1\} = \left\{ (t,x) \mid t = 2-\sqrt{1+x^2} \right\},
 \end{equation*}
 while for $\vert x \vert > 1$, $\operatorname{graph} (f)$ lies above $H$. It follows that $\tau(p) = 1$ and that any segment from a point in $\operatorname{graph} (f) \cap H$ to $p$ is a geodesic realizing $\tau$ (see Figure~\ref{fig:exam1}). Combining this with the fact that maximizing geodesics cannot branch, we deduce that the curve $\gamma \colon s \mapsto (s,0)$ has the following behavior:
 \begin{itemize}
  \item For $0 < s < 1$, $\gamma(s) \in \I$ and $\gamma$ is the unique $\tau$-realizing curve.
  \item For $s = 1$, $\gamma(1) \not\in \I$ and there are infinitely many $\tau$-realizers (including $\gamma$).
  \item For $s > 1$, $\gamma(s) \not\in \I$, $\gamma$ is not a $\tau$-realizer, and there exist two $\tau$-realizers (one in the left and one in the right half-plane).
 \end{itemize}
 In the last case, the $\tau$-realizers are segments whose endpoint on $\operatorname{graph} (f)$ can be found by intersecting $\operatorname{graph} (f)$ with an appropriate hyperboloid centered at $\gamma(s)$ (the one of largest radius that still intersects $\operatorname{graph} (f)$).
\end{example}

Based on the previous example, we construct one where $\mathcal{C}$ is not closed.

\begin{example} \label{exam2}
 Consider $f$ as before and $\tilde f$ given by
 \begin{equation*}
  \tilde f(x) := \begin{cases}
             \frac{f(a_n x+ b_n)}{a_n} &\text{if } -\frac{1}{n} \leq x < -\frac{1}{n+1}, \\
             0 &\text{if } x \geq 0.
            \end{cases}
 \end{equation*}
 where
 \begin{equation*}
  a_n = 4n(n+1), \qquad b_n = 4n+2,
 \end{equation*}
 and $M := \{ (t,x) \in \R^{1,1} \mid t > \tilde f(x) \}$. Then for $\gamma \colon s \mapsto (s,0)$, the following holds:
 \begin{itemize}
  \item $\gamma(s) \in \I = M \setminus \mathcal{C}$ for all $s > 0$.
  \item $\gamma(s) \in \overline{\mathcal{C}}$ for all $s > 0$. In particular, $\gamma(s)$ can be approximated by the sequence of points $p_n := (s,x_n)$ with $x_n < 0$ such that $\tilde f(x_n) = \frac{1}{a_n}$ is a local maximum, and for all $n$ large enough, $p_n \not\in \I$.
 \end{itemize}
 Thus the cosmological cut locus $\mathcal{C}$ is not closed. See also Figure \ref{fig:exam2}. Note that Theorem~\ref{thm:main} is applicable to this spacetime, so it has $\tau_V \in C^1$.
\end{example}

\begin{figure}
\centering
\begin{minipage}{.5\textwidth}
  \centering
  \includegraphics[width=\textwidth]{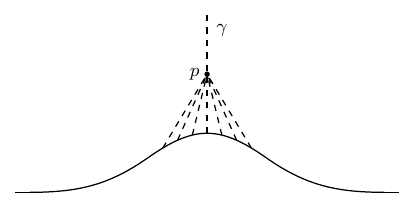}
  \caption{The graph of $f$ as in Example~\ref{exam1}, together with the vertical line $\gamma$, the point $p$, and multiple geodesics that realize $\tau$ at $p$ (dotted).}
  \label{fig:exam1}
\end{minipage}%
\begin{minipage}{.5\textwidth}
  \centering
  \includegraphics[width=\textwidth]{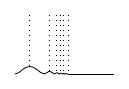}
  \caption{The graph of $\tilde f$ as in Example \ref{exam2}, together with examples of sequences $p_n \not\in \I$ converging (towards the right) to the vertical $t$-axis.}
  \label{fig:exam2}
\end{minipage}
\end{figure}

\section*{Declarations}

No data was generated or processed during this study. The author declares no conflicts of interest. This research was funded by the Austrian Science Fund (FWF) [Grant DOI 10.55776/EFP6]. For open access purposes, the author has applied a CC-BY public copyright license to any author accepted manuscript version arising from this submission. This version of the article has been accepted for publication after peer review, but is not the Version of Record and does not reflect post-acceptance improvements, or any corrections. The Version of Record is available online at: \\ http://dx.doi.org/10.1007/s11005-026-02112-5

\bibliographystyle{abbrv}
\bibliography{mybib.bib}

\end{document}